\documentclass[12pt, a4paper]{article}

\usepackage[utf8]{inputenc}
\usepackage[T1]{fontenc}
\usepackage[english]{babel}
\usepackage{amsmath, amsfonts, amssymb, amsthm}
\usepackage{mathtools}
\usepackage{geometry}
\usepackage{booktabs}
\usepackage{physics}
\usepackage{graphicx}
\usepackage{multirow}
\usepackage{chemfig}
\usepackage{float}
\usepackage{arydshln}
\usepackage{cite}
\usetikzlibrary{calc}

\newtheorem{theorem}{Theorem}
\newtheorem*{theorem*}{Theorem}
\newtheorem{corollary}{Corollary}[theorem]

\title{Mirror Symmetry of the NMR Spectrum and the Connection with the Structure of Spin Hamiltonian Matrix Representations}
\author{Dmitry A. Cheshkov and Dmitry O. Sinitsyn}
\date{\today}

\begin{document}

\maketitle

\begin{abstract}
This work provides a comprehensive theoretical framework for understanding the symmetry properties of High-Resolution NMR spectra.
We analyze the conditions under which a spectrum exhibits mirror symmetry (palindromicity).
We demonstrate that such symmetry can arise from two distinct mechanisms: (1) the direct geometric bisymmetry of the Hamiltonian matrix in a generalized canonical basis (typical for balanced systems like $A_nB_n$ or $A_nX_n$), and (2) a more fundamental property of topological isospectrality (similarity) under parameter exchange induced by the internal symmetry of the spin system, which applies even when the matrix lacks geometric symmetry (as observed in $AA'BB'$ systems).
\end{abstract}

\tableofcontents
\newpage

\section{General Formulas}

The Hamiltonian of a nuclear spin system in a strong magnetic field (high-resolution NMR) is represented as the sum of the Zeeman interaction and spin-spin interactions:

\begin{equation}
    \hat{H} = \hat{H}_Z + \hat{H}_{SS} = \hat{H}_Z + (\hat{H}_J + \hat{H}_D)
\end{equation}

Where the components are defined as follows:

\begin{itemize}
    \item \textbf{Zeeman Interaction ($\hat{H}_Z$):} Interaction of spins with the external magnetic field $B_0$.
    \begin{equation}
        \hat{H}_{Z}=\sum_{i=1}^{n}\nu_{i}\hat{I}_{iz}
    \end{equation}
    Here $\nu_i$ is the resonance frequency of the $i$-th nucleus, and $\hat{I}_{iz}$ is the operator of the z-component of the spin angular momentum.
    \item \textbf{Spin-Spin Interaction ($\hat{H}_{SS}$):} The sum of isotropic scalar ($J$) and secular dipolar ($D$) contributions.
    In the secular approximation (high field), only the terms commuting with total $I_z$ are retained:
    \begin{equation}
        \hat{H}_{SS} = \sum_{i<j} J_{ij} (\hat{\vec{I}}_{i} \cdot \hat{\vec{I}}_{j}) + \sum_{i<j} d_{ij}(3\hat{I}_{iz}\hat{I}_{jz}-\hat{\vec{I}}_{i} \cdot \hat{\vec{I}}_{j})
    \end{equation}
    where $J_{ij}$ is the scalar coupling constant and $d_{ij}$ is the residual dipolar coupling constant (if present).
\end{itemize}

\section{Canonical Basis and Matrix Symmetries}

It is important to distinguish between the symmetry of a specific spin system and the general symmetry properties of the Hamiltonian matrix representation.
In this section, we define the \textit{Canonical Basis}~$-$~a general-purpose basis applicable to any spin system.
We then demonstrate that in this basis, the Hamiltonian matrices possess inherent symmetries (persymmetry and bisymmetry) that stem solely from the algebra of spin operators.
For the matrix representation, we employ the basis of simple products of single-spin functions (\textit{Product functions basis}).
Each state $|\psi_k \rangle$ of a system of $N$ spins (where $I=1/2$) is defined by the set of magnetic quantum numbers $m_i = \pm 1/2$:

\begin{equation}
    |\psi_k \rangle = |m_1^{(k)}, m_2^{(k)}, \dots, m_N^{(k)} \rangle
\end{equation}

\subsection{Basis Ordering and Structure}
The basis functions are ordered to satisfy two fundamental conditions simultaneously:

\begin{enumerate}
    \item \textbf{$M_z$-Sorting:} The states are arranged by the descending order of the total spin projection $M_z = \sum m_i$.
    This ensures that the Hamiltonian matrix has a block-diagonal structure.
    \item \textbf{Spin-Inversion Centrosymmetry:} Within the $M_z$ ordering, the sequence is chosen such that the entire basis is centrosymmetric with respect to the global spin inversion operator $\hat{P}$ (where $\hat{P}$ flips all spins $m_i \to -m_i$).
\end{enumerate}

\subsection{Mathematical Definition of the Symmetry}
This centrosymmetric property implies a rigorous mapping between the physical operation of spin inversion and the index numbers of the basis functions.
If the dimension of the Hilbert space is $D = 2^N$, then for any $k$-th basis function:

\begin{equation}
    \hat{P} |\psi_k \rangle = |\psi_{D + 1 - k} \rangle
\end{equation}

This relationship links the algebraic properties of the spin operators directly to the persymmetry (symmetry with respect to the anti-diagonal) of the Hamiltonian matrix.

\subsection{Zeeman Interaction: Anti-Persymmetry}

In the product basis, the operator $\hat{H}_Z$ is strictly diagonal since the basis functions are eigenfunctions of $\hat{I}_{iz}$.
\subsubsection{Matrix Element}
The energy of state $| \mathbf{a} \rangle$ is the sum of resonance frequencies weighted by spin projections:
\begin{equation}
    \langle \mathbf{a} | \hat{H}_Z | \mathbf{a} \rangle = \sum_{i=1}^N \nu_i m_i^{(a)}
\end{equation}

\subsubsection{Symmetry Property}
Applying the inversion operator $\hat{P}$ flips all signs $m_i \to -m_i$, reversing the sign of the total Zeeman energy.
Due to the centrosymmetric basis property ($k \leftrightarrow D+1-k$), this implies that the Zeeman matrix is antisymmetric with respect to the anti-diagonal (anti-persymmetric):
\begin{equation}
    (H_Z)_{ii} = -(H_Z)_{D+1-i, D+1-i}
\end{equation}

\subsection{Spin-Spin Interaction: Bisymmetry}

\subsubsection{General Matrix Element Formula}
The full matrix element between states $|\mathbf{a}\rangle$ and $|\mathbf{b}\rangle$ is given by explicitly summing over all pairs:

\begin{equation}
\langle \mathbf{a} | \hat{H}_{SS} | \mathbf{b} \rangle = \sum_{k < l}^N \Bigg[ 
\underbrace{\delta_{\mathbf{a},\mathbf{b}} (J_{kl} + 2d_{kl}) m_k^{(a)} m_l^{(a)}}_{\text{Diagonal}} 
\;+\;
\underbrace{\frac{1}{2}(J_{kl} - d_{kl}) \hat{\Delta}_{kl}^{(\mathbf{a},\mathbf{b})}}_{\text{Off-diagonal}} 
\Bigg]
\end{equation}

Where the auxiliary symbols are defined as follows:

\begin{itemize}
    \item $\delta_{\mathbf{a},\mathbf{b}}$ (\textbf{Small Delta}) is the Kronecker delta, which selects the diagonal elements:
    \begin{equation}
        \delta_{\mathbf{a},\mathbf{b}} = \begin{cases} 1, & \text{if } |\mathbf{a}\rangle = |\mathbf{b}\rangle \\ 0, & \text{if } |\mathbf{a}\rangle \neq |\mathbf{b}\rangle \end{cases}
    \end{equation}
    
    \item $\hat{\Delta}_{kl}^{(\mathbf{a},\mathbf{b})}$ (\textbf{Large Delta}) is the selection filter for flip-flop transitions.
    It equals 1 if and only if the states differ exactly by the exchange of spins $k$ and $l$ (where $m_k \neq m_l$):
    \begin{equation}
        \hat{\Delta}_{kl}^{(\mathbf{a},\mathbf{b})} = \begin{cases} 1, & \text{if } |\mathbf{b}\rangle \in \{ \hat{I}_k^+ \hat{I}_l^- |\mathbf{a}\rangle, \: \hat{I}_k^- \hat{I}_l^+ |\mathbf{a}\rangle \} \\ 0, & \text{otherwise} \end{cases}
    \end{equation}
\end{itemize}

\vspace{1em}
\noindent \textit{Note:} While the secular Hamiltonian includes both scalar ($J$) and residual dipolar ($d$) couplings with different coefficients (see Eq. 8), both tensors share the same spatial symmetry properties defined by the molecular geometry.
For brevity, in the subsequent symmetry analysis, we will refer to the generalized coupling matrix simply as the \textit{$J$-coupling matrix}, implying that it effectively incorporates all spatial interaction constants.

\subsubsection{Bisymmetry Property}
The value of the spin-spin interaction matrix element is \textit{invariant} under global spin inversion $\hat{P}$:
\begin{equation}
    \langle \mathbf{a} | \hat{P}^\dagger \hat{H}_{SS} \hat{P} | \mathbf{b} \rangle = \langle \mathbf{a} | \hat{H}_{SS} | \mathbf{b} \rangle
\end{equation}
\textbf{Conclusion:} The $\hat{H}_{SS}$ matrix is symmetric with respect to both the main diagonal (Hermitian) and the anti-diagonal (persymmetric), making it a \textit{bisymmetric matrix}.
This property is universal.

\section{Elementary Spectral Transformations}

\subsection{Spectrum Translation via Frequency Offset}
In experimental NMR, changing the transmitter offset corresponds to a rigid translation of the spectrum.
\subsubsection{Hamiltonian Transformation}
Adding a constant offset $\Omega$ to all resonance frequencies ($\nu_i \to \nu_i + \Omega$) modifies the Hamiltonian:
\begin{equation}
    \hat{H}' = \hat{H} + \Omega \hat{I}_z^{tot}
\end{equation}
Since $[\hat{H}, \hat{I}_z^{tot}] = 0$ in the secular approximation, the eigenstates remain unchanged, but the energy eigenvalues are shifted linearly: $E'_n = E_n + \Omega M_z^{(n)}$.
\subsubsection{Shift of Spectral Lines}
For observable single-quantum transitions ($\Delta M_z = -1$), the frequency shift is:
\begin{equation}
    \omega'_{mn} = (E'_m - E'_n) = (E_m - E_n) + \Omega (M_z^{(m)} - M_z^{(n)}) = \omega_{mn} - \Omega
\end{equation}
Since the shift is identical for all lines, the internal multiplet structure (determined by $J$-couplings) remains invariant.
Thus, the offset operation performs a pure translation of the spectral pattern.

\subsection{Frequency Inversion and Spectral Reflection}
It is a fundamental property that inverting the signs of all resonance frequencies leads to a mirror reflection of the spectrum.
This reflection applies not only to the positions of the multiplets but also to their internal fine structure.
\subsubsection{Transformation of the Hamiltonian}
Let $\hat{H}'$ be the Hamiltonian where all resonance frequencies are inverted ($\nu_i \to -\nu_i$).
This is equivalent to:
\begin{equation}
    \hat{H}' = -\hat{H}_Z + \hat{H}_{SS}
\end{equation}
Applying the spin inversion operator $\hat{P}$:
\begin{equation}
    \hat{P}^\dagger \hat{H}' \hat{P} = -(-\hat{H}_Z) + \hat{H}_{SS} = \hat{H}_{original}
\end{equation}
Thus, $\hat{H}'$ and $\hat{H}_{original}$ are unitarily similar and share the \textit{same set of energy eigenvalues}.
This implies that the set of energy intervals (which define the magnitude of $J$-coupling splittings) remains numerically invariant.
\subsubsection{Inversion of Selection Rules}
Although the energy levels are the same, the observable transitions change due to the transformation of the eigenstates: $|\tilde{n}\rangle = \hat{P}|n\rangle$.
The observable signal depends on the transition matrix elements of the lowering operator $\hat{I}^-$.
The operator $\hat{P}$ transforms the lowering operator into the raising operator:
\begin{equation}
    \hat{P}^\dagger \hat{I}^- \hat{P} = -\hat{I}^+
\end{equation}
Consequently, a transition allowed in the original system at frequency $\omega = E_f - E_i$ corresponds to a transition in the inverted system associated with the reverse process ($\Delta M = +1$), appearing at frequency $-\omega$.
\subsubsection{Reflection of Fine Structure}
Crucially, this mechanism reflects the \textit{entire topology} of the spectrum.
One might intuitively expect that since the scalar couplings $J_{ij}$ are not inverted, the multiplets would simply shift positions while retaining their original shape.
However, the unitary transformation $\hat{P}$ affects the mixing coefficients of the wavefunctions.
As a result:
\begin{itemize}
    \item The center of the multiplet moves from $\nu$ to $-\nu$.
    \item The internal structure is mirrored: a transition that was on the "right" side of the multiplet (e.g., $\nu + J/2$) maps to the "left" side relative to the new center ($-\nu - J/2$).
    \item Intensity distortions, such as the "roof effect" (where inner lines are stronger), are also strictly reflected.
\end{itemize}
Thus, $\hat{H}'(-\nu)$ generates a spectrum $S(-\omega)$ which is a perfect mirror image of $S(\omega)$ down to the finest detail.

\subsection{Frequency Inversion and Magnet Reversal}
Consider the transformation where all resonance frequencies are inverted ($\nu_i \to -\nu_i$), which is mathematically equivalent to $\hat{H}' = -\hat{H}_Z + \hat{H}_{SS}$.
Physically, this corresponds to reversing the external magnetic field ($B_0 \to -B_0$).
\begin{itemize}
    \item \textbf{Eigenvalues:} $\hat{H}'$ is unitarily equivalent to $\hat{H}$ via the spin inversion operator $\hat{P}$, so they share the same set of eigenvalues.
    \item \textbf{Selection Rules:} $\hat{P}$ transforms the lowering operator $\hat{I}^-$ into $\hat{I}^+$, effectively reversing the "direction" of transitions ($|n\rangle \to |m\rangle$ becomes $|\tilde{m}\rangle \to |\tilde{n}\rangle$).
\end{itemize}
\textbf{Conclusion:} The spectrum in a negative field is the mirror image of the spectrum in a positive field: $S(\omega) \to S(-\omega)$.

\section{Spectral Symmetry: The Ideal Case ($[\hat{H}, \hat{Q}] = 0$)}

\paragraph{The Conflict of Symmetries}
The total Hamiltonian $\hat{H} = \hat{H}_Z + \hat{H}_{SS}$ generally yields an asymmetric spectrum because its components transform differently under $\hat{P}$: $\hat{H}_Z$ is anti-persymmetric, while $\hat{H}_{SS}$ is bisymmetric.
\paragraph{Restoration via Generalized Parity}
Symmetry can be restored if the anti-persymmetry of the Zeeman term is compensated by the structural symmetry of the parameters.
We introduce the \textit{Generalized Parity Operator} $\hat{Q} = \hat{P} \times \hat{\Pi}$, where $\hat{\Pi}$ is the permutation that reverses the order of spins ($i \leftrightarrow N+1-i$).
The fundamental algebraic condition for symmetry is $[\hat{H}, \hat{Q}] = 0$.

\subsection{Connection: Why Commutation Implies Symmetry}
Why does the algebraic condition $[\hat{H}, \hat{Q}] = 0$ force the spectrum to be a palindrome?
\subsubsection{The Indistinguishability Argument}
\begin{enumerate}
    \item \textbf{Magnet Reversal Effect:} Changing the sign of the magnetic field reflects the spectrum: $S(\omega) \to S(-\omega)$ (Section 5.3).
    \item \textbf{Structural Symmetry:} If the system parameters satisfy the conditions for $[\hat{H}, \hat{Q}] = 0$, then the state of the system in a reversed field ($-\boldsymbol{\nu}$) is structurally identical to the original state ($\boldsymbol{\nu}$), differing only in the spins order (operator $\hat{\Pi}$).
    \item \textbf{Conclusion:} Since the physics is invariant under relabeling, the spectrum of the "reversed magnet" system must be identical to the spectrum of the "original" system:
    \begin{equation}
        S(-\omega) = S(\omega)
    \end{equation}
\end{enumerate}
The commutation relation $[\hat{H}, \hat{Q}] = 0$ is the rigorous operator statement that the system is invariant under the combined operation of field inversion and relabeling.

\subsection{The Generalized Canonical Basis ($\hat{Q}$-basis)}

To make the spectral symmetry manifest in the matrix form, we introduce a reordered basis adapted to the combined operator $\hat{Q}$.
In this basis, the pairing of states is defined by the generalized parity operator $\hat{Q} = \hat{P}\hat{\Pi}$.

\subsection*{Mechanism of Symmetry Restoration}
In a centered spectrum ($\sum \nu_i = 0$), the resonance frequencies obey $\nu_k = -\nu_{N+1-k}$.
Since the paired states in this basis are structural mirror images of each other, every spin contribution $+\nu_i$ in one state is perfectly compensated by a contribution $-\nu_i$ in its partner.

Consider a 3-spin system. The state $| +-- \rangle$ ($m_1=+, m_2=-, m_3=-$) transforms into its symmetric partner as follows:
\begin{equation*}
    | +-- \rangle \xrightarrow{\hat{\Pi}} | --+ \rangle \xrightarrow{\hat{P}} | ++- \rangle
\end{equation*}

This combined operation guarantees that if the resonance frequencies are antisymmetric relative to the center, the paired states possess identical Zeeman energies ($E = E'$).
This transforms the Zeeman matrix for these \textit{Shell} states from anti-persymmetric to \textit{persymmetric}, matching the symmetry of the spin-spin interaction.

\subsubsection{Structure and Symmetry of the Zero-Quantum Block ($M_z=0$)}
For subspaces with $M_z \neq 0$, the operator $\hat{Q}$ simply maps the entire block $+M$ to the block $-M$.
However, in systems with an even number of spins ($N \ge 2$), the Zero-Quantum subspace ($M_z = 0$) is mapped onto itself.
It is only within this block that \textit{self-conjugate states} (states invariant under $\hat{Q}$) can arise.
To manage this complexity, the basis is organized into a \textit{Nested "Shell-Core" Structure}, which leads to a specific hybrid symmetry of the Zeeman interaction.

\begin{enumerate}
    \item \textbf{The Shell (Transverse Pairs):} 
    States that are \textit{not} invariant under $\hat{Q}$ form pairs $|\psi_a\rangle \leftrightarrow |\psi_b\rangle$.
    \begin{itemize}
        \item \textit{Symmetry Implication:} As a consequence of the restoration mechanism described above, the total Zeeman energies of the pair are strictly equal ($E_a = E_b$).
        This equality transforms the Zeeman matrix relation from anti-persymmetry ($H_{ij} = -H_{ji}$) to \textit{persymmetry} ($H_{ij} = H_{ji}$), resolving the conflict with the symmetric $J$-coupling term.
    \end{itemize}

    \item \textbf{The Core (Self-Conjugate States):} 
    States that are invariant under $\hat{Q}$ lie in the exact center of the basis ($M_z = 0$).
    They are mapped onto themselves (or within the small core block) by $\hat{Q}$, meaning they lack a distinct "frequency-compensating" partner.
    \begin{itemize}
        \item \textit{Mechanism of Persistence:} For these states, the pairing is defined solely by spin inversion $\hat{P}$, which reverses all spins without swapping nuclei.
        This operation simply flips the sign of the total energy: $E \to -E$.
        \item \textit{Symmetry Implication:} Since there is no partner with equal energy to restore balance, the condition $E_i = -E_j$ persists.
        Therefore, the Zeeman matrix in the "Core" region remains \textit{anti-persymmetric}, preventing the total Hamiltonian from becoming globally bisymmetric.
        \item \textit{Algebraic Consistency:} Crucially, this violation of geometric bisymmetry does not imply a breakdown of the fundamental commutation relation $[\hat{H}, \hat{Q}] = 0$ (for balanced systems).
        The apparent conflict is an artifact of the matrix representation: for self-conjugate states, $\hat{Q}$ acts via spin inversion $\hat{P}$ (flipping Zeeman energy signs) rather than by swapping degenerate partners.
        Thus, algebraic symmetry is preserved despite the hybrid geometric structure.
    \end{itemize}
\end{enumerate}

\textbf{Conclusion:} The total Zeeman matrix in the $\hat{Q}$-basis possesses a \textit{hybrid symmetry}: it is persymmetric in the outer shell but anti-persymmetric in the inner core.
This hybrid nature prevents the total Hamiltonian $\hat{H} = \hat{H}_Z + \hat{H}_{SS}$ from being globally bisymmetric, even if $\hat{H}_{SS}$ is perfect. Crucially, this geometric asymmetry of the matrix representation
does not violate the algebraic symmetry condition.

\subsection{Rigorous Mathematical Proofs}

\subsubsection{Proof via Unitary Transformation}
The NMR spectrum $S(\omega)$ is the Fourier transform of the correlation function.
Applying the unitary transformation $\hat{Q}$ to the trace (noting that $\hat{Q}^\dagger \hat{I}_x \hat{Q} = -\hat{I}_x$ and the trace is invariant under unitary similarity $\text{Tr}(A) = \text{Tr}(\hat{Q}^\dagger A \hat{Q})$):
\begin{equation}
    \text{Corr}'(t) = \text{Tr} \left( e^{-i\hat{H}t} (-\hat{I}_x) e^{i\hat{H}t} (-\hat{I}_x) \right) = \text{Corr}(t)
\end{equation}
Since $\hat{Q}$ involves spin inversion, it effectively maps $S(\omega) \to S(-\omega)$.
If the correlation function is invariant, the spectrum is symmetric.

\subsubsection{Proof via Method of Moments}
The symmetry can also be proven by showing that all odd Van Vleck moments vanish.
The $n$-th moment $M_n$ is defined as:
\begin{equation}
    M_n = \int_{-\infty}^{\infty} \omega^n I(\omega) d\omega \propto \text{Tr} \{ [\hat{H}, [\hat{H}, \dots [\hat{H}, \hat{I}_x]\dots]] \hat{I}_x \}
\end{equation}
Applying the transformation $\hat{P}$ maps $\hat{H} \to \hat{H}'(-\nu)$.
Since the spectrum of the inverted system is the mirror image of the original, $I(\omega) = I(-\omega)$.
Consequently, all odd moments $M_{1}, M_{3}, \dots$ must be identically zero.

\subsection{Generalized Symmetry (Reflect and Shift)}
If the system possesses structural symmetry ($J_{ij}$ are centrosymmetric) but the transmitter offset is not set to zero (i.e., $\nu_i \neq -\nu_{N+1-i}$), the spectrum is still a palindrome, but centered at $\omega_0 \neq 0$.
The generalized symmetry condition is:
\begin{equation}
    S(2\omega_0 - \omega) = S(\omega)
\end{equation}
This implies that the mirror image of the spectrum can be superimposed onto the original spectrum by a rigid linear translation.

\subsection{Invariance of the Characteristic Polynomial}

The condition of unitary similarity discussed above has a fundamental algebraic consequence concerning the characteristic polynomial of the Hamiltonian, defined as $P(\lambda) = \det(\hat{H} - \lambda \hat{E})$.

\begin{theorem}[Spectral Enantiomerism]
If two distinct Hamiltonians possess identical characteristic polynomials, then their spectra either coincide or are mirror reflections of each other (spectral ``enantiomers'').
\end{theorem}

\begin{proof}
The identity of characteristic polynomials, $P_1(\lambda) = P_2(\lambda)$, implies that the two Hamiltonians share the exact same set of eigenvalues (energy levels). Mathematically, this means the matrices are \textit{isospectral}. Consequently, the set of all possible transition frequencies $\omega_{ij} = E_i - E_j$ is identical for both systems.

The distinction between "identical" and "mirror-image" spectra arises from the properties of the eigenvectors, which determine the transition intensities:
\begin{enumerate}
    \item \textbf{Identical Spectra:} If the unitary transformation $\hat{U}$ linking the Hamiltonians ($\hat{H}' = \hat{U}^\dagger \hat{H} \hat{U}$) preserves the selection rules (e.g., commutes with total spin $I_z$), the transition intensities remain mapped to the same frequencies.
    \item \textbf{Spectral Enantiomers:} If the transformation inverts the selection rules (e.g., maps transitions $\Delta M = -1$ to $\Delta M = +1$, as seen with the operator $\hat{Q}$ involving spin inversion), the intensities are mapped to the reflected frequencies ($-\omega$).
\end{enumerate}
\end{proof}

\section{Beyond Geometric Symmetry: The Case of $AA'BB'$ ($AA'XX'$) Spin System}

This section presents a comprehensive analysis of the origin of mirror symmetry in the NMR spectra of $AA'BB'$ ($AA'XX'$) spin systems [\citen{Pople1959}, p.~116; \citen{Corio1966}, p.~376; \citen{Gunther2013}, p.~192], based on the generalized parity operator $\hat{Q} = \hat{P}\hat{\Pi}$ framework.
We demonstrate that the observed spectral symmetry is an inherent algebraic property dictated by the structure of the spin Hamiltonian.
By decomposing the Hamiltonian into symmetry-distinct invariant subspaces, we identify two independent mechanisms that ensure spectral symmetry despite the geometric asymmetry of the $J$-coupling network ($J_{AA'} \neq J_{BB'}$).

\subsection{System Description}
The analyzed $AA'BB'$ spin system is defined by the following parameters:
\[
\renewcommand{\arraystretch}{1.5}
\begin{array}{r|cccc}
 Spins & A & A' & B & B' \\
 \hline
 \nu_i & \nu_0 + \Delta/2 & \nu_0 + \Delta/2 & \nu_0 - \Delta/2 & \nu_0 - \Delta/2 \\
 J\text{-}couplings & & J_{AA'} & J_{AB} & J_{AB'} \\
 \ & & & J_{AB'} & J_{AB} \\
 & & & & J_{BB'}
\end{array}
\]

\begin{itemize}
    \item \textbf{Spins:} Four spins labeled $A, A', B, B'$.
    \item \textbf{Resonance Frequencies:} Symmetric distribution relative to the center $\nu_0$:
    \[ \nu_A = \nu_{A'} = \nu_0 + \Delta/2, \quad \nu_B = \nu_{B'} = \nu_0 - \Delta/2 \]
    \item \textbf{\textit{J}-Coupling Network:} The scalar coupling matrix $\boldsymbol{J}$ is defined by the unique coupling constants ($J_{AB} \neq J_{AB'}$ and $J_{AA'} \neq J_{BB'}$) as follows:
    \[
    \boldsymbol{J} = 
    \begin{pmatrix}
     0 & J_{AA'} & J_{AB} & J_{AB'} \\
     J_{AA'} & 0 & J_{AB'} & J_{AB} \\
     J_{AB} & J_{AB'} & 0 & J_{BB'} \\
     J_{AB'} & J_{AB} & J_{BB'} & 0
    \end{pmatrix}
    \]
\end{itemize}

\subsection{Representations in P-Basis and Q-Basis}

\subsubsection{The P-Basis (P-Ordered Product Basis)}
The P-basis is partitioned into subspaces according to the total angular momentum projection $M_z$.
Crucially, the ordering is defined such that the entire basis sequence forms a palindrome with respect to the collective spin inversion operator $\hat{P}$.
This means the second half of the basis is the state-by-state spin inversion of the first half (in reverse order).

\vspace{1em}
{\footnotesize
\noindent\textit{Spin Product basis functions are conveniently defined using the binary representation of integer numbers (from $0$ to $2^N-1$), grouped by Hamming weight.
For four spins, the P-ordered Product basis takes the form:}
\[
\resizebox{\textwidth}{!}{
\renewcommand{\arraystretch}{1.5}
\setlength{\tabcolsep}{4pt}
\newcommand{\stP}[1]{\makebox[3.6em][c]{$#1$}}
\newcommand{\arrP}[1]{\makebox[3.6em][c]{\scriptsize$#1$}}
$
\begin{array}{c|c|c|c|c}
 M_z=+2 & M_z=+1 & M_z=0 & M_z=-1 & M_z=-2 \\
 \hline
 \stP{0} & 
 \stP{1}\stP{2}\stP{4}\stP{8} & 
 \stP{6}\stP{5}\stP{3}\;\makebox[1em]{$\boldsymbol{||}$}\;\stP{12}\stP{10}\stP{9} & 
 \stP{7}\stP{11}\stP{13}\stP{14} & 
 \stP{15} 
 \\
 \arrP{\uparrow\uparrow\uparrow\uparrow} &
 \arrP{\uparrow\uparrow\uparrow\downarrow}\arrP{\uparrow\uparrow\downarrow\uparrow}\arrP{\uparrow\downarrow\uparrow\uparrow}\arrP{\downarrow\uparrow\uparrow\uparrow} &
 \arrP{\uparrow\downarrow\downarrow\uparrow}\arrP{\uparrow\downarrow\uparrow\downarrow}\arrP{\uparrow\uparrow\downarrow\downarrow}\;\makebox[1em]{}\;\arrP{\downarrow\downarrow\uparrow\uparrow}\arrP{\downarrow\uparrow\downarrow\uparrow}\arrP{\downarrow\uparrow\uparrow\downarrow} &
 \arrP{\uparrow\downarrow\downarrow\downarrow}\arrP{\downarrow\uparrow\downarrow\downarrow}\arrP{\downarrow\downarrow\uparrow\downarrow}\arrP{\downarrow\downarrow\downarrow\uparrow} &
 \arrP{\downarrow\downarrow\downarrow\downarrow}
\end{array}
$
}
\]
\textit{The symbol $\boldsymbol{||}$ demarcates the center of the basis and relates the second half to the first via mirror reflection with spin inversion.}
}

\begin{itemize}
    \item \textbf{Zeeman Hamiltonian ($\hat{H}_Z$):} The matrix is diagonal.
    Since the spectrum is centered, the elements exhibit \textit{anti-persymmetry}:
    \[ (H_Z)_{i,i} = -(H_Z)_{N-i+1, N-i+1} \]
    \item \textbf{Spin-Spin Hamiltonian ($\hat{H}_{SS}$):} The matrix is \textit{bisymmetric} (symmetric with respect to both the main and anti-diagonals).
\end{itemize}

\subsubsection{The Q-Basis (Generalized Parity Basis)}

The Q-basis is constructed to align with the generalized parity operator $\hat{Q} = \hat{P}\hat{\Pi}$.
In the first half, it coincides with the P-basis; in the second half, the spin order of the P-basis states is reversed.
This additional spin order reversal results in some basis functions within the $M_z = 0$ subblock being mapped onto themselves under the action of $\hat{Q}$, thereby forming the kernel of the Q-basis.

\vspace{1em}
{\footnotesize
\noindent\textit{For four spins, the Q-ordered Product basis takes the form:}
\[
\resizebox{\textwidth}{!}{
\renewcommand{\arraystretch}{1.2}
\setlength{\tabcolsep}{4pt}
\newcommand{\stQ}[1]{\makebox[3.6em][c]{$#1$}}
\newcommand{\arrQ}[1]{\makebox[3.6em][c]{\scriptsize$#1$}}
$
\begin{array}{c|c|c|c|c}
 M_z=+2 & M_z=+1 & M_z=0 & M_z=-1 & M_z=-2 \\
 \hline
 \stQ{0} & 
 \stQ{1}\stQ{2}\stQ{4}\stQ{8} & 
 \stQ{6}\;\stQ{5}\stQ{3}\;\makebox[1em]{$\boldsymbol{||}$}\;\stQ{12}\stQ{10}\;\stQ{9} & 
 \stQ{14}\stQ{13}\stQ{11}\stQ{7} & 
 \stQ{15} 
 \\
 \arrQ{\uparrow\uparrow\uparrow\uparrow} &
 \arrQ{\uparrow\uparrow\uparrow\downarrow}\arrQ{\uparrow\uparrow\downarrow\uparrow}\arrQ{\uparrow\downarrow\uparrow\uparrow}\arrQ{\downarrow\uparrow\uparrow\uparrow} &
 \arrQ{\uparrow\downarrow\downarrow\uparrow}\;
 \arrQ{\uparrow\downarrow\uparrow\downarrow}\arrQ{\uparrow\uparrow\downarrow\downarrow}\;\makebox[1em]{}\;\arrQ{\downarrow\downarrow\uparrow\uparrow}\arrQ{\downarrow\uparrow\downarrow\uparrow}
 \;\arrQ{\downarrow\uparrow\uparrow\downarrow} &
 \arrQ{\downarrow\downarrow\downarrow\uparrow}\arrQ{\downarrow\downarrow\uparrow\downarrow}\arrQ{\downarrow\uparrow\downarrow\downarrow}\arrQ{\uparrow\downarrow\downarrow\downarrow} &
 \arrQ{\downarrow\downarrow\downarrow\downarrow}
 \\[-0.5ex] 
 \multicolumn{2}{c}{
    \underbrace{
        \phantom{\arrQ{\uparrow}}\quad\phantom{\arrQ{\uparrow}\arrQ{\uparrow}\arrQ{\uparrow}\arrQ{\uparrow}}
    }_{\text{Outer Shell}}
 } & 
 \multicolumn{1}{c}{
    \phantom{\arrQ{\uparrow}}\;
    \underbrace{
        \phantom{\arrQ{\uparrow}\arrQ{\uparrow}}\;\makebox[1em]{}\;\phantom{\arrQ{\uparrow}\arrQ{\uparrow}}
    }_{\text{Core}}
    \;\phantom{\arrQ{\uparrow}}
 } & 
 \multicolumn{2}{c}{
    \underbrace{
        \phantom{\arrQ{\uparrow}\arrQ{\uparrow}\arrQ{\uparrow}\arrQ{\uparrow}}\quad\phantom{\arrQ{\uparrow}}
    }_{\text{Outer Shell}}
 } 
\end{array}
$
}
\]
\textit{The symbol $\boldsymbol{||}$ demarcates the center of the basis.}
}

\paragraph{The Outer Shells ($M_z \neq 0$)}
These subblocks correspond to $\hat{Q}$-basis states for $M_z = \pm 2$ and $M_z = \pm 1$.
\begin{itemize}
    \item \textbf{Structure of $\hat{H}_{SS}$:} The global bisymmetry is lost. The "lower" block is the geometric reflection of the "upper" block, but with a \textit{complete parameter substitution} $J_{AA'} \leftrightarrow J_{BB'}$.
    \item \textbf{Structure of $\hat{H}_Z$:} The matrix is diagonal and exhibits global \textit{persymmetry} due to the resonance frequency balance ($\nu_A = -\nu_B$).
\end{itemize}

\paragraph{The Zero-Quantum Block ($M_z = 0$)}
In this central subblock, the Q-basis identifies the \textit{self-conjugate Core}.
\begin{enumerate}
    \item \textbf{Inner Shells:}
    Paired states surrounding the core that map into each other under $\hat{Q}$.
    \begin{itemize}
        \item \textit{Spin-Spin Term:} Retains \textit{bisymmetry}, but creates a symmetry conflict due to mixing with the Core.
        \item \textit{Zeeman Term:} Exhibits \textit{persymmetry}.
    \end{itemize}
    \item \textbf{The Core:}
    Self-conjugate states grouped in the center.
    \begin{itemize}
        \item \textit{Spin-Spin Term:} Retains \textit{bisymmetry}.
        \item \textit{Zeeman Term:} Retains \textit{anti-persymmetry} ($H_{ii} = -H_{N-i+1, N-i+1}$).
    \end{itemize}
\end{enumerate}

\subsection{Conjugation by Operator $\hat{Q}$}
The action of the operator $\hat{Q}$ on the Hamiltonian corresponds to the transformation $\hat{H}' = \hat{Q}^\dagger \hat{H} \hat{Q}$.

\subsubsection{Action in the P-Basis}
In the P-Ordered Product Basis, the spin-spin Hamiltonian matrix is naturally bisymmetric. Consequently, the geometric aspect of the transformation maps the matrix onto itself. The transformation manifests purely algebraically as a parameter swap ($J_{AA'} \leftrightarrow J_{BB'}$).

\subsubsection{Action in the Q-Basis}
\begin{enumerate}
    \item \textbf{Outer Shells ($M_z \neq 0$): Global Reflection} 
    For the spatially separated blocks, the transformation acts as a geometric reflection across the anti-diagonal.
    \item \textbf{Zero-Quantum Block ($M_z = 0$):} 
    The matrix is geometrically invariant.
    \begin{itemize}
        \item \textbf{Inner Shells:} Manifests as an \textit{in-place parameter swap} ($J_{AA'} \leftrightarrow J_{BB'}$).
        \item \textbf{The Core:} Numerical values remain strictly invariant because the matrix elements depend exclusively on the symmetric parameter sums ($\Sigma J_{intra}$).
    \end{itemize}
\end{enumerate}

\subsection{Symmetrized Representation}

Transitioning to a symmetry-adapted basis via the unitary transformation $U$ block-diagonalizes the Hamiltonian into a symmetric subspace and an asymmetric subspace.

\vspace{1em}
\noindent
\begin{minipage}{\linewidth}
{\footnotesize
\noindent\textit{Basis functions of the $AA'BB'$ spin system, symmetrized according to the irreducible representations of the $C_2$ symmetry group:}
\begin{center}
\makebox[\textwidth][c]{%
    \setlength{\tabcolsep}{3pt}
    \renewcommand{\arraystretch}{1.1}
    \newcommand{\vStrut}{\rule[-1.2em]{0pt}{3.0em}}
    \newcommand{\hMz}[1]{{\fontsize{6}{7}\selectfont $M_z=#1$}}
    \newcommand{\hRow}[1]{{\fontsize{6}{7}\selectfont \textbf{#1}}}
    \begin{tabular}{l | c | c | c | c | c}
     & \hMz{+2} & \hMz{+1} & \hMz{0} & \hMz{-1} & \hMz{-2} \\
     \hline
     \hRow{IR A} \textit{ \fontsize{6}{7}\selectfont (sym)} \quad & 
     \vStrut $|0\rangle$ & 
     \quad $\frac{|1\rangle+|2\rangle}{\sqrt{2}} \quad \frac{|4\rangle+|8\rangle}{\sqrt{2}}$ \quad & 
     \quad $\frac{|6\rangle+|9\rangle}{\sqrt{2}} \quad |3\rangle \quad |12\rangle \quad \frac{|5\rangle+|10\rangle}{\sqrt{2}}$ \quad & 
     \quad $\frac{|13\rangle+|14\rangle}{\sqrt{2}} \quad \frac{|7\rangle+|11\rangle}{\sqrt{2}}$ \quad & 
     \quad $|15\rangle$ \quad
     \\
     \hRow{IR B} \textit{ \fontsize{6}{7}\selectfont (asym)} \quad & 
     \vStrut & 
     \quad $\frac{|1\rangle-|2\rangle}{\sqrt{2}} \quad \frac{|4\rangle-|8\rangle}{\sqrt{2}}$ \quad & 
     \quad $\frac{|6\rangle-|9\rangle}{\sqrt{2}} \quad \frac{|5\rangle-|10\rangle}{\sqrt{2}}$ \quad & 
     \quad $\frac{|13\rangle-|14\rangle}{\sqrt{2}} \quad \frac{|7\rangle-|11\rangle}{\sqrt{2}}$ \quad & 
    \end{tabular}
}
\end{center}
}
\end{minipage}
\vspace{1em}

\subsubsection{The Symmetric Subspace}
The symmetric subspace corresponds to states that transform according to the trivial representation. The Hamiltonian block $\hat{H}_{sym}$ depends exclusively on the sum of $J_{AA'}$ and $J_{BB'}$, so $[\hat{H}_{sym}, \hat{Q}] = 0$.

\textbf{Structure of the Zero-Quantum Block:}
The explicit form of this $4 \times 4$ doubled Hamiltonian block illustrates the extensive connectivity between basis states:

\begin{equation}
\label{eq:sym_zero_block}
\scalebox{0.9}{
$
\begin{pmatrix}
 -\frac{1}{2}\Sigma J_{intra} - \Delta J_{inter} & \sqrt{2} J_{AB} & \sqrt{2} J_{AB} & \Sigma J_{intra} \\
 \sqrt{2} J_{AB} & \frac{1}{2}\Sigma J_{intra} - \Sigma J_{inter} + 2\Delta & 0 & \sqrt{2} J_{AB'} \\
 \sqrt{2} J_{AB} & 0 & \frac{1}{2}\Sigma J_{intra} - \Sigma J_{inter} - 2\Delta & \sqrt{2} J_{AB'} \\
 \Sigma J_{intra} & \sqrt{2} J_{AB'} & \sqrt{2} J_{AB'} & -\frac{1}{2}\Sigma J_{intra} + \Delta J_{inter}
\end{pmatrix}
$
}
\end{equation}

\subsubsection{The Asymmetric Subspace}
To illustrate the origin of isospectrality, we decompose the Hamiltonian matrix in the asymmetric subspace. The decomposition of the doubled Hamiltonian matrix $2\hat{H}_{asym}$ is given by:

\begin{equation}
\label{eq:decomposition}
\scalebox{0.87}{
$
\setlength{\arraycolsep}{2pt}
\begin{gathered}
2\hat{H}_{asym} = 
\underbrace{
-\frac{1}{2}\Sigma J_{intra} \mathbf{I}_6
}_{\text{Scalar Shift}}
+
\underbrace{
\begin{pmatrix}
 \Delta & 0 & 0 & 0 & 0 & 0 \\
 0 & -\Delta & 0 & 0 & 0 & 0 \\
 0 & 0 & 0 & 0 & 0 & 0 \\
 0 & 0 & 0 & 0 & 0 & 0 \\
 0 & 0 & 0 & 0 & -\Delta & 0 \\
 0 & 0 & 0 & 0 & 0 & \Delta
\end{pmatrix}
}_{\text{Zeeman Term}}
+
\underbrace{
\begin{pmatrix}
 0 & \Delta J_{inter} & 0 & 0 & 0 & 0 
 \\
 \Delta J_{inter} & 0 & 0 & 0 & 0 & 0 \\
 0 & 0 & -\Delta J_{inter} & 0 & 0 & 0 \\
 0 & 0 & 0 & \Delta J_{inter} & 0 & 0 \\
 0 & 0 & 0 & 0 & 0 & \Delta J_{inter} \\
 0 & 0 & 0 & 0 & \Delta J_{inter} & 0
\end{pmatrix}
}_{\text{Invariant Coupling}}
\\[10pt]
+
\underbrace{
\begin{pmatrix}
 \Delta J_{intra} & 0 & 0 & 0 & 0 & 0 \\
 0 & -\Delta J_{intra} & 0 & 0 & 0 & 0 \\
 0 & 0 & 0 & -\Delta J_{intra} 
 & 0 & 0 \\
 0 & 0 & -\Delta J_{intra} & 0 & 0 & 0 \\
 0 & 0 & 0 & 0 & \Delta J_{intra} & 0 \\
 0 & 0 & 0 & 0 & 0 & -\Delta J_{intra}
\end{pmatrix}
}_{\text{Asymmetric Part}}
\end{gathered}
$
}
\end{equation}

The \textit{Asymmetric Part} depends on $\Delta J_{intra} = J_{AA'} - J_{BB'}$. Since the parameter exchange swaps $J_{AA'}$ and $J_{BB'}$, this term changes its sign under $\hat{Q}$.
This results in the transformed matrix:

\begin{equation}
\resizebox{1.1\textwidth}{!}{
$
\underbrace{
\begin{pmatrix}
 d_+ & \Delta J_{inter} & 0 & 0 & 0 & 0 \\
 \Delta J_{inter} & -d_+ & 0 & 0 & 0 & 0 \\
 0 & 0 & -\Delta J_{inter} & -\Delta J_{intra} & 0 & 0 \\
 0 & 0 & -\Delta J_{intra} & \Delta J_{inter} & 0 & 0 \\
 0 & 0 & 0 & 0 & -d_- & \Delta J_{inter} \\
 0 & 0 & 0 & 0 & \Delta J_{inter} & d_-
\end{pmatrix}
}_{2\hat{H}_{asym}}
\xrightarrow{\hat{Q}}
\underbrace{
\begin{pmatrix}
 d_- & \Delta J_{inter} & 
 0 & 0 & 0 & 0 \\
 \Delta J_{inter} & -d_- & 0 & 0 & 0 & 0 \\
 0 & 0 & -\Delta J_{inter} & +\Delta J_{intra} & 0 & 0 \\
 0 & 0 & +\Delta J_{intra} & \Delta J_{inter} & 0 & 0 \\
 0 & 0 & 0 & 0 & -d_+ & \Delta J_{inter} \\
 0 & 0 & 0 & 0 & \Delta J_{inter} & d_+
\end{pmatrix}
}_{2\hat{Q}^\dagger \hat{H}_{asym} \hat{Q}}
$
}
\end{equation}

\noindent where $d_{\pm} = \Delta \pm \Delta J_{intra}$.

\paragraph{Structure of Subblocks and Spectral Invariance}
It can be observed that all subblocks are traceless symmetric matrices of dimension $2\times2$. The eigenvalues for such matrices are equal and opposite in sign.

\vspace{-1em}
\begin{center}
\renewcommand{\arraystretch}{1.3} 
\setlength{\tabcolsep}{3pt}
\resizebox{\textwidth}{!}{
\begin{tabular}{c | cc | cc | cc}
 \multicolumn{1}{c|}{} & \multicolumn{2}{c|}{$M_z = +1$} & \multicolumn{2}{c|}{$M_z = 0$} & \multicolumn{2}{c}{$M_z = -1$} \\ 
 \hline
 \rule{0pt}{6ex}$\hat{H}_{asym}$ & 
 \multicolumn{2}{c|}{$\begin{pmatrix} d_+ & \Delta J_{inter} \\ \Delta J_{inter} & -d_+ \end{pmatrix}$} & 
 \multicolumn{2}{c|}{$\begin{pmatrix} -\Delta J_{inter} & -\Delta J_{intra} \\ -\Delta J_{intra} & \Delta J_{inter} \end{pmatrix}$} & 
 \multicolumn{2}{c}{$\begin{pmatrix} -d_- & \Delta J_{inter} \\ \Delta J_{inter} & d_- \end{pmatrix}$} \\
 
 \rule[-3ex]{0pt}{9ex}$\tan 2\theta$ & 
 \multicolumn{2}{c|}{$\displaystyle \tan 2\theta_+ = \frac{\Delta J_{inter}}{d_+}$} & 
 \multicolumn{2}{c|}{$\displaystyle \tan 2\theta_0 = \frac{\Delta J_{intra}}{\Delta J_{inter}}$} & 
 \multicolumn{2}{c}{$\displaystyle \tan 2\theta_- = \frac{\Delta J_{inter}}{d_-}$} \\
 
 \rule[-6ex]{0pt}{12ex}$E$ & 
 $\begin{gathered} +E_+ \\[0.5ex] +\sqrt{\Delta J_{inter}^2 + d_+^2} \end{gathered}$ & 
 $\begin{gathered} -E_+ \\[0.5ex] -\sqrt{\Delta J_{inter}^2 + d_+^2} \end{gathered}$ & 
 $\begin{gathered} -E_0 \\[0.5ex] -\sqrt{\Delta J_{inter}^2 + \Delta J_{intra}^2} \end{gathered}$ & 
 $\begin{gathered} +E_0 \\[0.5ex] +\sqrt{\Delta J_{inter}^2 + \Delta J_{intra}^2} \end{gathered}$ & 
 $\begin{gathered} -E_- \\[0.5ex] -\sqrt{\Delta J_{inter}^2 + d_-^2} \end{gathered}$ & 
 $\begin{gathered} +E_- \\[0.5ex] +\sqrt{\Delta J_{inter}^2 + d_-^2} \end{gathered}$ \\
 
 \rule[-3ex]{0pt}{7ex}$\psi$ & 
 $\begin{pmatrix} \cos\theta_+ \\ \sin\theta_+ \end{pmatrix}$ & 
 $\begin{pmatrix} -\sin\theta_+ \\ \cos\theta_+ \end{pmatrix}$ & 
 $\begin{pmatrix} \cos\theta_0 \\ \sin\theta_0 \end{pmatrix}$ & 
 $\begin{pmatrix} -\sin\theta_0 \\ \cos\theta_0 \end{pmatrix}$ & 
 $\begin{pmatrix} \cos\theta_- \\ -\sin\theta_- \end{pmatrix}$ & 
 $\begin{pmatrix} \sin\theta_- \\ \cos\theta_- \end{pmatrix}$
\end{tabular}
}
\end{center}
\vspace{1em}

\noindent
\begin{minipage}{\linewidth}
Doubled perturbation operator, line frequencies, and transition moments:

\vspace{-1.5em}
\begin{center}
\small
\renewcommand{\arraystretch}{1.5}
\resizebox{\textwidth}{!}{
\begin{tabular}{c|c|r@{\,$\to$\,}l|c} 

 \multicolumn{1}{c|}{Perturbation operator} & 
 \multicolumn{1}{c|}{Frequency} & 
 \multicolumn{2}{c|}{Transition} & 
 \multicolumn{1}{c}{Doubled transition moment} \\ \hline
 
 \multirow{4}{*}{\begin{tabular}{c} $M_z=+1 \to M_z=0$ \\[1.5ex] $2\Sigma \hat{I}_{x,asym} = \begin{pmatrix} -1 & 1 \\ 1 & 1 \end{pmatrix}$ \end{tabular}} 
 & $\phantom{-}f_1$ & $+E_+$ & $-E_0$ & $\cos\theta_0(-\cos\theta_+ + \sin\theta_+) + \sin\theta_0(\cos\theta_+ + \sin\theta_+)$ \\
 & $\phantom{-}f_2$ & $+E_+$ & $+E_0$ & $\cos\theta_0(\cos\theta_+ + \sin\theta_+) + \sin\theta_0(\cos\theta_+ - \sin\theta_+)$ \\
 & $-f_2$           & $-E_+$ & $-E_0$ & $\cos\theta_0(\cos\theta_+ + \sin\theta_+) + 
 \sin\theta_0(\cos\theta_+ - \sin\theta_+)$ \\
 & $-f_1$           & $-E_+$ & $+E_0$ & $-[\cos\theta_0(-\cos\theta_+ + \sin\theta_+) + \sin\theta_0(\cos\theta_+ + \sin\theta_+)]$ \\ \hline
 
 \multirow{4}{*}{\begin{tabular}{c} $M_z=0 \to M_z=-1$ \\[1.5ex] $2\Sigma \hat{I}_{x,asym} = \begin{pmatrix} 1 & -1 \\ -1 & -1 \end{pmatrix}$ \end{tabular}} 
 & $\phantom{-}f_3$ & $-E_0$ & $-E_-$ & $\cos\theta_0(\cos\theta_- + \sin\theta_-) + \sin\theta_0(-\cos\theta_- + \sin\theta_-)$ \\
 & $\phantom{-}f_4$ & $-E_0$ & $+E_-$ & $-[\cos\theta_0(\cos\theta_- - \sin\theta_-) + \sin\theta_0(\cos\theta_- + \sin\theta_-)]$ \\
 & $-f_4$           & $+E_0$ & $-E_-$ & $-[\cos\theta_0(\cos\theta_- - \sin\theta_-) 
 + \sin\theta_0(\cos\theta_- + \sin\theta_-)]$ \\
 & $-f_3$           & $+E_0$ & $+E_-$ & $-[\cos\theta_0(\cos\theta_- + \sin\theta_-) + \sin\theta_0(-\cos\theta_- + \sin\theta_-)]$ 
\end{tabular}
}
\end{center}
\end{minipage}
\vspace{0.5em}

The subspectrum of the asymmetric subspace consists of four symmetrically located doublets.
The action of the operator $\hat{Q}$ leads to the permutation of pairs of doublets.

\noindent
\begin{minipage}{\linewidth}
After the action of $\hat{Q}$:

\vspace{-1em}
\begin{center}
\renewcommand{\arraystretch}{1.3}
\setlength{\tabcolsep}{3pt}
\resizebox{\textwidth}{!}{%
\begin{tabular}{c | cc | cc |
 cc}
 \multicolumn{1}{c|}{} & \multicolumn{2}{c|}{$M_z = +1$} & \multicolumn{2}{c|}{$M_z = 0$} & \multicolumn{2}{c}{$M_z = -1$} \\ 
 \hline
 
 \rule{0pt}{6ex}$\hat{H}'_{asym}$ & 
 \multicolumn{2}{c|}{$\begin{pmatrix} d_- & \Delta J_{inter} \\ \Delta J_{inter} & -d_- \end{pmatrix}$} & 
 \multicolumn{2}{c|}{$\begin{pmatrix} -\Delta J_{inter} & +\Delta J_{intra} \\ +\Delta J_{intra} & \Delta J_{inter} \end{pmatrix}$} & 
 \multicolumn{2}{c}{$\begin{pmatrix} -d_+ & \Delta J_{inter} \\ \Delta J_{inter} & d_+ \end{pmatrix}$} \\
 
 \rule[-3ex]{0pt}{9ex}$\tan 2\theta'$ & 
 \multicolumn{2}{c|}{$\displaystyle \tan 2\theta_- = \frac{\Delta J_{inter}}{d_-}$} & 
 \multicolumn{2}{c|}{$\displaystyle \tan 2\theta_0 = \frac{\Delta J_{intra}}{\Delta J_{inter}}$} & 
 \multicolumn{2}{c}{$\displaystyle \tan 2\theta_+ = \frac{\Delta J_{inter}}{d_+}$} \\
 
 \rule[-6ex]{0pt}{12ex}$E'$ & 
 $\begin{gathered} +E_- \\[0.5ex] +\sqrt{\Delta J_{inter}^2 + 
 d_-^2} \end{gathered}$ & 
 $\begin{gathered} -E_- \\[0.5ex] -\sqrt{\Delta J_{inter}^2 + d_-^2} \end{gathered}$ & 
 $\begin{gathered} -E_0 \\[0.5ex] -\sqrt{\Delta J_{inter}^2 + \Delta J_{intra}^2} \end{gathered}$ & 
 $\begin{gathered} +E_0 \\[0.5ex] +\sqrt{\Delta J_{inter}^2 + \Delta J_{intra}^2} \end{gathered}$ & 
 $\begin{gathered} -E_+ \\[0.5ex] -\sqrt{\Delta J_{inter}^2 + d_+^2} \end{gathered}$ & 
 $\begin{gathered} +E_+ \\[0.5ex] +\sqrt{\Delta J_{inter}^2 + d_+^2} \end{gathered}$ \\
 
 \rule[-3ex]{0pt}{7ex}$\psi'$ & 
 $\begin{pmatrix} \cos\theta_- \\ \sin\theta_- \end{pmatrix}$ & 
 $\begin{pmatrix} -\sin\theta_- \\ \cos\theta_- \end{pmatrix}$ & 
 $\begin{pmatrix} \cos\theta_0 \\ -\sin\theta_0 \end{pmatrix}$ & 
 $\begin{pmatrix} \sin\theta_0 \\ \cos\theta_0 \end{pmatrix}$ & 
 $\begin{pmatrix} \sin\theta_+ \\ \cos\theta_+ \end{pmatrix}$ & 
 $\begin{pmatrix} \cos\theta_+ \\ -\sin\theta_+ \end{pmatrix}$
\end{tabular}
}
\end{center}
\end{minipage}
\vspace{1em}

\noindent
\begin{minipage}{\linewidth}
Doubled perturbation operator, line frequencies, and transition moments ($\hat{Q}^\dagger \Sigma \hat{I}_{x,asym} \hat{Q} = \Sigma \hat{I}_{x,asym}$):

\vspace{-1.5em}
\begin{center}
\small
\renewcommand{\arraystretch}{1.5}
\resizebox{\textwidth}{!}{
\begin{tabular}{c|c|r@{\,$\to$\,}l|c} 

 \multicolumn{1}{c|}{Perturbation operator} & 
 \multicolumn{1}{c|}{Frequency} & 
 \multicolumn{2}{c|}{Transition} & 
 \multicolumn{1}{c}{Doubled transition moment} \\ \hline
 
 \multirow{4}{*}{\begin{tabular}{c} $M_z=+1 \to M_z=0$ \\[1.5ex] $2\Sigma \hat{I}_{x,asym} = \begin{pmatrix} -1 & 1 \\ 1 & 1 \end{pmatrix}$ \end{tabular}} 
 & $-f_4$           & $+E_-$ & $-E_0$ & $-[\cos\theta_0(\cos\theta_- - \sin\theta_-) + \sin\theta_0(\cos\theta_- + \sin\theta_-)]$ \\
 & $\phantom{-}f_3$ & $+E_-$ & $+E_0$ & $\cos\theta_0(\cos\theta_- + \sin\theta_-) + \sin\theta_0(-\cos\theta_- + \sin\theta_-)$ \\
 & $-f_3$           & $-E_-$ & $-E_0$ & $\cos\theta_0(\cos\theta_- + \sin\theta_-) + \sin\theta_0(-\cos\theta_- + \sin\theta_-)$ \\
 & $\phantom{-}f_4$ & $-E_-$ & $+E_0$ & $\cos\theta_0(\cos\theta_- - \sin\theta_-) + \sin\theta_0(\cos\theta_- + \sin\theta_-)]$ \\ \hline
 
 \multirow{4}{*}{\begin{tabular}{c} $M_z=0 \to M_z=-1$ \\[1.5ex] $2\Sigma \hat{I}_{x,asym} = \begin{pmatrix} 1 & -1 \\ -1 & -1 \end{pmatrix}$ \end{tabular}} 
 & $\phantom{-}f_2$ & $-E_0$ & $-E_+$ & $\cos\theta_0(\cos\theta_+ + \sin\theta_+) + \sin\theta_0(\cos\theta_+ - \sin\theta_+)$ \\
 & $-f_1$           & $-E_0$ & $+E_+$ & $\cos\theta_0(-\cos\theta_+ + \sin\theta_+) + \sin\theta_0(\cos\theta_+ + \sin\theta_+)$ \\
 & $\phantom{-}f_1$ & $+E_0$ & $-E_+$ & $\cos\theta_0(-\cos\theta_+ + \sin\theta_+) + \sin\theta_0(\cos\theta_+ + \sin\theta_+)$ \\
 & $-f_2$           & $+E_0$ & $+E_+$ & $-[\cos\theta_0(\cos\theta_+ + \sin\theta_+) + \sin\theta_0(\cos\theta_+ - \sin\theta_+)]$ 
\end{tabular}
}
\end{center}
\end{minipage}
\vspace{0.5em}

The action of the operator $\hat{Q}$ leads to the permutation of pairs of doublets, those due to transitions $\{M_z = +1 \to M_z = 0\}$ and transitions $\{M_z = 0 \to M_z = -1\}$.
Now the doublets with frequencies $\pm f_3$ and $\pm f_4$ correspond to transitions $M_z = +1 \to M_z = 0$, and the doublets with frequencies $\pm f_1$ and $\pm f_2$ correspond to transitions $M_z = 0 \to M_z = -1$.
The line intensities are preserved and do not change when the blocks are swapped.

\paragraph{Static Nature of Spectral Symmetry}
The explicit Hamiltonian subblocks structure demonstrates that the observed spectral symmetry in the asymmetric subspace is a \textit{static} geometric property.
Since all matrices are inherently traceless, the energy levels are strictly constrained to appear in symmetric pairs $\pm E$.
Consequently, the parameter asymmetry $J_{AA'} \neq J_{BB'}$ (or $\Delta J_{intra} \neq 0$) affects only the magnitude of the energy splitting but is mathematically incapable of breaking the symmetry of the spectrum.

This specific structure leads to a distinct topology of transitions compared to the symmetric subspace.
In the symmetric subspace, the spectral mirror symmetry arises from the correspondence between transitions belonging to symmetrically located blocks (pairing the transitions from the upper part of the energy diagram with those from the lower part).
In contrast, in the asymmetric subspace, the symmetric transitions are generated \textit{locally}.
Specifically, the transitions between blocks $M_z=+1 \to M_z=0$ (and similarly $M_z=0 \to M_z=-1$) independently form symmetric frequency subsets due to the inherent $\pm E$ structure of the participating levels.
Thus, the palindromic structure of transitions is guaranteed a priori by the algebraic form of the Hamiltonian in the symmetrized basis.

\subsection{Summary: The Four Pillars of $AA'BB'$ Symmetry Analysis}

Our analysis confirms that the observation of spectral symmetry relies on four critical components.
Violation of any of these conditions severs the link between the Hamiltonian's algebraic structure and the spectrum's geometric symmetry.
\begin{enumerate}
    \item \textbf{Correct Spin Indexing (Topology of $J$):}
    The nuclei must be indexed (numbered) such that the $J$-coupling matrix exhibits the required topological invariance or balanced reflection relative to the anti-diagonal.
    Arbitrary indexing destroys the structure of the permutation operator $\hat{\Pi}$, making the generalized parity operator $\hat{Q}$ physically meaningless for the system.
    \item \textbf{Frequency Centering ($\sum \nu_i = 0$):}
    Defining resonance frequencies relative to the spectral center of gravity is strictly necessary to render the Zeeman matrix perfectly anti-persymmetric ($H_{ii} = -H_{jj}$).
    In the "Shell" regions, this condition ensures the exact compensation of energies ($+\nu$ and $-\nu$), converting the Zeeman anti-symmetry into the persymmetry required to resolve the conflict with the symmetric $J$-coupling term.
    \item \textbf{Canonical Basis Sorting:}
    The basis functions must be rigorously ordered by $M_z$ and, crucially, arranged to be centrosymmetric with respect to the spin inversion operator $\hat{P}$.
    This ensures that the physical operation of spin inversion maps exactly to the geometric reflection of the matrix indices ($k \leftrightarrow D+1-k$).
    \item \textbf{Correct Symmetrizing Transformation ($\hat{U}$):}
    The unitary transition from the product basis to the symmetry-adapted basis is essential to factorize the Hamiltonian.
    It reveals the distinct algebraic nature of the irreducible blocks:
    \begin{itemize}
        \item The Symmetric block commutes with $\hat{Q}$ ($[\mathbf{H}_{sym}, \hat{Q}] = 0$), implying geometric invariance.
        \item The Asymmetric block exhibits \textit{Static Geometric Symmetry} (Isospectrality). While it does not commute with $\hat{Q}$ in the strict operator sense due to parameter swap, its subblocks are inherently traceless symmetric matrices, which enforces the spectral invariance ($\pm E$) under the parameter reflection.
    \end{itemize}
\end{enumerate}

\vspace{1em}
\noindent The theoretical explanation of spectral symmetry is the result of the synergy between these factors:
\begin{equation}
    \text{Spectral Symmetry} = \text{Topology}(J) + \text{Centering}(\nu) + \text{Basis}(\hat{P}) + \text{Algebra}(\hat{Q}).
\end{equation}

\subsection{Comparison with Magnetic Equivalence $A_nB_n$ ($A_nX_n$)}
It is important to distinguish the $AA'BB'$ case from systems with magnetic equivalence, such as $A_nB_n$ ($A_nX_n)$.
In $A_nB_n$ systems, the intra-group couplings $J_{AA}$ and $J_{BB}$ do not affect the transition energies (they disappear from the relevant parts of the Hamiltonian).
Thus, the condition for spectral symmetry is satisfied trivially. In contrast, in $AA'BB'$ systems, $J_{AA'}$ and $J_{BB'}$ actively contribute to the spectrum, and symmetry
is preserved only due to the specific algebraic structure (cancellation of signs in squared differences) described above.

\section{The General Theorem of Spectral Symmetry}

Based on the preceding analysis, we can formulate the final necessary and sufficient condition for the mirror symmetry of an NMR spectrum.
We introduce the concept of \textit{reflected frequencies} relative to the spectral center of gravity $\nu_0 = \frac{1}{n} \sum \nu_i$.

\begin{theorem}[General Isospectrality and Structural Equivalence]
The NMR spectrum $S(\omega)$ is mirror-symmetric with respect to its center \textit{if and only if} the Hamiltonian satisfies the condition of \textit{frequency reflection isospectrality}:
\begin{equation}
    \label{eq:freq_iso}
    \text{Spec}(\hat{H}(\boldsymbol{\nu}, \boldsymbol{J})) = \text{Spec}(\hat{H}(2\nu_0 - \boldsymbol{\nu}, \boldsymbol{J}))
\end{equation}
subject to the constraint that the unitary transformation linking these Hamiltonians corresponds to a permutation of nuclei.

\vspace{0.5em}
\noindent \textbf{Structural Equivalence:}
This fundamental condition is equivalent to the requirement that there exists at least one specific spin order (permutation $1 \dots N$) which simultaneously ensures:
\begin{enumerate}
    \item \textit{Frequency Balance:} The resonance frequencies are distributed centrosymmetrically in this order ($\nu_i - \nu_0 = -(\nu_{N+1-i} - \nu_0)$).\footnote{\textit{In systems containing large groups of chemically equivalent nuclei, e.g., $[A'B']_n$, the requirement of sorting by resonance frequencies is necessary but not sufficient. Due to the permutation symmetry within equivalent groups, there exist multiple spin orderings that satisfy the frequency sorting condition. However, as matrix analysis demonstrates, the condition of $J$-coupling matrix symmetry is sensitive to the relative ordering of spins within these chemically equivalent groups. Therefore, for such spin systems, the criterion is strictly defined as follows: Mirror symmetry exists if there is at least one specific permutation within the chemically equivalent sets that results in a symmetric or "balanced" $J$-coupling matrix.}}
    \item \textit{Reflection Invariance of the Interaction:} Under this ordering, the frequency reflection is physically equivalent to the reflection of the interaction matrix. Thus, condition (\ref{eq:freq_iso}) reduces to:
    \begin{equation}
        \text{Spec}(\hat{H}(\boldsymbol{\nu}, \boldsymbol{J})) = \text{Spec}(\hat{H}(\boldsymbol{\nu}, \boldsymbol{J}^{refl}))
    \end{equation}
    where the reflected matrix is defined as $(\boldsymbol{J}^{refl})_{ij} = (\boldsymbol{J})_{N+1-j, N+1-i}$.
\end{enumerate}
\end{theorem}

\noindent \textit{Physical Interpretation:}
The theorem states that spectral symmetry arises if and only if the spin system admits a specific "palindromic" spin order. In this order, the system is indistinguishable (isospectral) from its own mirror image obtained by reversing the sequence of spin resonance frequencies ($1 \leftrightarrow N, 2 \leftrightarrow N-1, \dots$). This means that swapping the $J$-couplings across the "center" of the index sequence (e.g., exchanging $J_{1,2}$ with $J_{N,N-1}$) must leave the spectrum invariant.

\begin{corollary}[Types of Symmetry]
This theorem encompasses both mechanisms discussed in this work:
\begin{itemize}
    \item \textit{Geometric Symmetry:} If $\boldsymbol{J} = \boldsymbol{J}^{refl}$ (the matrix is explicitly bisymmetric), the Hamiltonian strictly commutes with the generalized parity operator ($[\hat{H}, \hat{Q}] = 0$), and the spectrum is identical (trivial isospectrality). This corresponds to systems like $A_nB_n$.
    \item \textit{Isospectral Symmetry:} If $\boldsymbol{J} \neq \boldsymbol{J}^{refl}$ but the Hamiltonians are isospectral due to algebraic properties (caused by symmetry of spin system), the spectrum is symmetric despite the violation of strict commutation. This corresponds to systems like $AA'BB'$.
\end{itemize}
\end{corollary}

\section{Spectral Asymmetry of Some Symmetric $[A'X']_n$ Spin Systems}

\begin{figure}[H]
    \centering
    $D_3$-symmetric $[A'X']_3$ spin system of 1,3,5-trifluorobenzene\\
    \par
    \vspace{1.0em}
    \begin{tikzpicture}[scale=0.94, baseline=(current bounding box.center)]
        \def\R{1.0}
        \coordinate (Center) at (0,0);
        \coordinate (C1) at (330:\R); \coordinate (C2) at (270:\R);
        \coordinate (C3) at (210:\R); \coordinate (C4) at (150:\R);
        \coordinate (C5) at (90:\R);  \coordinate (C6) at (30:\R);
        
        \draw[thick] (C1) -- (C2) -- (C3) -- (C4) -- (C5) -- (C6) -- cycle;
        \draw[thick] ($(C2)!0.13!(Center)$) -- ($(C3)!0.13!(Center)$);
        \draw[thick] ($(C4)!0.13!(Center)$) -- ($(C5)!0.13!(Center)$);
        \draw[thick] ($(C6)!0.13!(Center)$) -- ($(C1)!0.13!(Center)$);
        
        \draw[thick] (C1) -- ++(330:0.5) node[right] {H$^{A''}$};
        \draw[thick] (C2) -- ++(270:0.5) node[below] {F$^{X'}$};
        \draw[thick] (C3) -- ++(210:0.5) node[left] {H$^{A'}$};
        \draw[thick] (C4) -- ++(150:0.5) node[left] {F$^X$};
        \draw[thick] (C5) -- ++(90:0.5) node[above] {H$^A$};
        \draw[thick] (C6) -- ++(30:0.5) node[right] {F$^{X''}$};
    \end{tikzpicture}
    \hspace{2mm}
    \centering
        \renewcommand{\arraystretch}{1.5}
        \setlength{\tabcolsep}{6pt}
        \begin{tabular}{r|ccc:ccc}
            \textit{Spins} & A & A$'$ & A$''$ & X & X$'$ & X$''$ \\ \hline
            \textit{$\nu_i$} & $\nu_A$ & $\nu_A$ & $\nu_A$ & $\nu_X$ & $\nu_X$ & $\nu_X$ \\ \hdashline 
            \multirow[t]{5}{*}{\textit{J-couplings}} 
              &  & $J_{AA'}$ & $J_{AA'}$ & $J_{AX}$ & $J_{AX'}$ & $J_{AX}$ \\
              &  &  & $J_{AA'}$ & $J_{AX}$ & $J_{AX}$ & $J_{AX'}$ \\
              &  &  &  & $J_{AX'}$ & $J_{AX}$ & $J_{AX}$ \\
            \hdashline
              &  &  &  &  & $J_{XX'}$ & $J_{XX'}$ \\
              &  &  &  &  &  & $J_{XX'}$ \\
        \end{tabular}
\end{figure}

\begin{figure}[H]
    \centering
    $C_{2\mathrm{v}}$-symmetric $[A'X']_4$ spin systems\\
    \par
    \vspace{1.0em}
    \begin{tikzpicture}[scale=1.45, baseline=(current bounding box.center)]
        \coordinate (O)  at (0, -0.8);
        \coordinate (BL) at (-1.0, -0.2); \coordinate (BR) at (1.0, -0.2);
        \coordinate (TL) at (-0.6, 0.6);  \coordinate (TR) at (0.6, 0.6);
        \draw[thick] (O) -- (BL) -- (TL) -- (TR) -- (BR) -- (O);
        \draw[thick] (TL) -- (TR);
        \node at (O) [fill=white, inner sep=1pt] {O};
        \draw[thick] (TL) -- ++(0, 0.50) node[above] {H$^X$};
        \draw[thick] (TL) -- ++(0,-0.50) node[below] {H$^{X'}$};
        \draw[thick] (TR) -- ++(0, 0.50) node[above] {H$^{X''}$};
        \draw[thick] (TR) -- ++(0,-0.50) node[below] {H$^{X'''}$};
        \draw[thick] (BL) -- ++(0, 0.50) node[above] {H$^A$};
        \draw[thick] (BL) -- ++(0,-0.50) node[below] {H$^{A'}$};
        \draw[thick] (BR) -- ++(0, 0.50) node[above] {H$^{A''}$};
        \draw[thick] (BR) -- ++(0,-0.50) node[below] {H$^{A'''}$};
    \end{tikzpicture}
    \hspace{2mm}
    \begin{tikzpicture}[scale=0.94, baseline=(current bounding box.center)]
        \coordinate (L_Center) at (-0.866, 0); \coordinate (R_Center) at (0.866, 0);
        \coordinate (S_Top) at (0, 0.5);   \coordinate (S_Bot) at (0, -0.5);
        \coordinate (R_Top) at (0.866, 1); \coordinate (R_SideT) at (1.732, 0.5);
        \coordinate (R_SideB) at (1.732, -0.5); \coordinate (R_Bot) at (0.866, -1);
        \coordinate (L_Top) at (-0.866, 1); \coordinate (L_SideT) at (-1.732, 0.5);
        \coordinate (L_SideB) at (-1.732, -0.5); \coordinate (L_Bot) at (-0.866, -1);
        \draw[thick] (L_Top)--(L_SideT)--(L_SideB)--(L_Bot)--(S_Bot)--(R_Bot)--(R_SideB)--(R_SideT)--(R_Top)--(S_Top)--cycle;
        \draw[thick] (S_Top)--(S_Bot);
        \draw[thick] ($(S_Top)!0.13!(R_Center)$) -- ($(S_Bot)!0.13!(R_Center)$);
        \draw[thick] ($(L_Top)!0.13!(L_Center)$) -- ($(L_SideT)!0.13!(L_Center)$);
        \draw[thick] ($(L_Bot)!0.13!(L_Center)$) -- ($(L_SideB)!0.13!(L_Center)$);
        \draw[thick] ($(R_Top)!0.13!(R_Center)$) -- ($(R_SideT)!0.13!(R_Center)$);
        \draw[thick] ($(R_Bot)!0.13!(R_Center)$) -- ($(R_SideB)!0.13!(R_Center)$);
        \draw[thick] (L_Top) -- ++(0, 0.6) node[above] {H$^A$};
        \draw[thick] (L_Bot) -- ++(0, -0.6) node[below] {H$^{A'}$};
        \draw[thick] (R_Top) -- ++(0, 0.6) node[above] {H$^{A''}$};
        \draw[thick] (R_Bot) -- ++(0, -0.6) node[below] {H$^{A'''}$};
        \draw[thick] (L_SideT) -- ++(150:0.6) node[left] {H$^X$};
        \draw[thick] (L_SideB) -- ++(210:0.6) node[left] {H$^{X'}$};
        \draw[thick] (R_SideT) -- ++(30:0.6)  node[right] {H$^{X''}$};
        \draw[thick] (R_SideB) -- ++(330:0.6) node[right] {H$^{X'''}$};
    \end{tikzpicture}
    \hspace{2mm}
    \begin{tikzpicture}[scale=0.94, baseline=(current bounding box.center)]
        \def\R{1.3} \coordinate (Center) at (0,0);
        \coordinate (C1) at (157.5:\R); \coordinate (C2) at (112.5:\R);
        \coordinate (C3) at (67.5:\R);  \coordinate (C4) at (22.5:\R);
        \coordinate (C5) at (337.5:\R); \coordinate (C6) at (292.5:\R);
        \coordinate (C7) at (247.5:\R); \coordinate (C8) at (202.5:\R);
        \draw[thick] (C1)--(C2)--(C3)--(C4)--(C5)--(C6)--(C7)--(C8)--cycle;
        \foreach \start/\end in {C2/C3, C4/C5, C6/C7, C8/C1} {
             \draw[thick] ($(\start)!0.1!(Center)$) -- ($(\end)!0.1!(Center)$);
        }
        \draw[thick] (C2) -- ++(112.5:0.4) node[above left]  {{H}$^A$};
        \draw[thick] (C3) -- ++(67.5:0.4)  node[above right] {{H}$^{A''}$};
        \draw[thick] (C7) -- ++(247.5:0.4) node[below left]  {{H}$^{A'}$};
        \draw[thick] (C6) -- ++(292.5:0.4) node[below right] {{H}$^{A'''}$};
        \draw[thick] (C1) -- ++(157.5:0.4) node[left]  {{F}$^X$};
        \draw[thick] (C8) -- ++(202.5:0.4) node[left]  {{F}$^{X'}$};
        \draw[thick] (C4) -- ++(22.5:0.4)  node[right] {{F}$^{X''}$};
        \draw[thick] (C5) -- ++(337.5:0.4) node[right] {{F}$^{X'''}$};
        \node at (Center) [font=\large] {2-};
    \end{tikzpicture}
    
    \par
    \vspace{0.5em}
        \renewcommand{\arraystretch}{1.4}
        \setlength{\tabcolsep}{4pt}
        \begin{tabular}{r|cccc:cccc}
            \textit{Spins} & A & A$'$ & A$''$ & A$'''$ & X & X$'$ & X$''$ & X$'''$ \\ \hline
            \textit{$\nu_i$} & $\nu_A$ & $\nu_A$ & $\nu_A$ & $\nu_A$ & $\nu_X$ & $\nu_X$ & $\nu_X$ & $\nu_X$ \\ \hdashline 
            \multirow[t]{7}{*}{\textit{J-couplings}} 
              &  & $J_{AA'}$ & $J_{AA''}$ & $J_{AA'''}$ & $J_{AX}$ & $J_{AX'}$ & $J_{AX''}$ & $J_{AX'''}$ \\
              &  &  & $J_{AA'''}$ & $J_{AA''}$ & $J_{AX'}$ & $J_{AX}$ & $J_{AX'''}$ & $J_{AX''}$ \\
              &  &  &  & $J_{AA'}$ & $J_{AX''}$ & $J_{AX'''}$ & $J_{AX}$ & $J_{AX'}$ \\
              &  &  &  &  & $J_{AX'''}$ & $J_{AX''}$ & $J_{AX'}$ & $J_{AX}$ \\ 
            \hdashline
              &  &  &  &  &  & $J_{XX'}$ & $J_{XX''}$ & $J_{XX'''}$ \\
              &  &  &  &  &  &  & $J_{XX'''}$ & $J_{XX''}$ \\
              &  &  &  &  &  &  &  & $J_{XX'}$ \\
        \end{tabular}
\end{figure}

\begin{figure}[H]
   \centering
   $D_4$-symmetric $[A'X']_4$ spin system\\
    \par
    \vspace{1.0em}
        \centering
        \begin{tikzpicture}[scale=0.94, baseline=(current bounding box.center)]
            \def\R{1.3} 
            \coordinate (Center) at (0,0);
            \coordinate (C1) at (157.5:\R); \coordinate (C2) at (112.5:\R);
            \coordinate (C3) at (67.5:\R);  \coordinate (C4) at (22.5:\R);
            \coordinate (C5) at (337.5:\R); \coordinate (C6) at (292.5:\R);
            \coordinate (C7) at (247.5:\R); \coordinate (C8) at (202.5:\R);
            \draw[thick] (C1)--(C2)--(C3)--(C4)--(C5)--(C6)--(C7)--(C8)--cycle;
            \foreach \start/\end in {C2/C3, C4/C5, C6/C7, C8/C1} {
                 \draw[thick] ($(\start)!0.1!(Center)$) -- ($(\end)!0.1!(Center)$);
            }
            \draw[thick] (C2) -- ++(112.5:0.4) node[above left] {H$^A$};
            \draw[thick] (C3) -- ++(67.5:0.4) node[above right] {F$^{X'''}$};
            \draw[thick] (C4) -- ++(22.5:0.4) node[right] {H$^{A'''}$};
            \draw[thick] (C5) -- ++(337.5:0.4) node[right] {F$^{X''}$};
            \draw[thick] (C6) -- ++(292.5:0.4) node[below right] {H$^{A''}$};
            \draw[thick] (C7) -- ++(247.5:0.4) node[below left] {F$^{X'}$};
            \draw[thick] (C8) -- ++(202.5:0.4) node[left] {H$^{A'}$};
            \draw[thick] (C1) -- ++(157.5:0.4) node[left] {F$^X$};
            \node at (Center) [font=\large] {2-};
        \end{tikzpicture}
        \renewcommand{\arraystretch}{1.4}
        \setlength{\tabcolsep}{4pt}
        \begin{tabular}{r|cccc:cccc}
            \textit{Spins} & A & A$'$ & A$''$ & A$'''$ & X & X$'$ & X$''$ & X$'''$ \\ \hline
            \textit{$\nu_i$} & $\nu_A$ & $\nu_A$ & $\nu_A$ & $\nu_A$ & $\nu_X$ & $\nu_X$ & $\nu_X$ & $\nu_X$ \\ \hdashline 
            \multirow[t]{7}{*}{\textit{J-couplings}} 
              &  & $J_{AA'}$ & $J_{AA''}$ & $J_{AA'}$ & $J_{AX}$ & $J_{AX'}$ & $J_{AX'}$ & $J_{AX}$ \\
              &  &  & $J_{AA'}$ & $J_{AA''}$ & $J_{AX}$ & $J_{AX}$ & $J_{AX'}$ & $J_{AX'}$ \\
              &  &  &  & $J_{AA'}$ & $J_{AX'}$ & $J_{AX}$ & $J_{AX}$ & $J_{AX'}$ \\
              &  &  &  &  & $J_{AX'}$ & $J_{AX'}$ & $J_{AX}$ & $J_{AX}$ \\ 
            \hdashline
              &  &  &  &  &  & $J_{XX'}$ & $J_{XX''}$ & $J_{XX'}$ \\
              &  &  &  &  &  &  & $J_{XX'}$ & $J_{XX''}$ \\
              &  &  &  &  &  &  &  & $J_{XX'}$ \\
        \end{tabular}
\end{figure}

The spin ordering was assigned as follows: 1) spins were grouped by chemical equivalence to satisfy the resonance frequency balance condition; 2) within chemically equivalent groups ($A$ and $X$), the order was chosen to be topologically mutually consistent, ensuring that $J_{AX} = J_{A'X'} = J_{A''X''} = J_{A'''X'''}$.
For 1,3,5-trifluorobenzene and the 1,3,5,7-tetrafluoro-substituted cyclooctatetraene dianion, this ordering coincides with the canonical topological order.
In all considered spin systems, the $J$-coupling matrices exhibit high symmetry, with the inter-group coupling blocks being symmetric with respect to the anti-diagonal (persymmetric).
However, in contrast to the $[AA'BB']$ case, due to the algebraic properties of the Hamiltonian, there is no balancing of topologically equivalent homonuclear coupling constant pairs in these systems. Specifically, the constant pairs $\{J_{AA'}, J_{XX'}\}$, $\{J_{AA''}, J_{XX''}\}$, and $\{J_{AA'''}, J_{XX'''}\}$ do not form balanced pairs.
Consequently, the spectra are not invariant under the permutation of constants within these pairs and do not possess mirror symmetry, despite the high symmetry of the spin systems themselves.
\par
\vspace{1.0em}
{\footnotesize
\textit{It is worth noting that in the product basis, the doubled off-diagonal elements of the spin Hamiltonian correspond to the spin-spin coupling constants. In systems with a single second-order symmetry element (e.g., the $[AA'BB']$ case), the symmetrized basis consists of sums and differences of pairs of product basis functions, leading to off-diagonal elements that contain sums and differences of spin-spin coupling constant pairs. However, in systems with higher symmetry, the symmetrized linear combinations involve a larger number of basis functions. As a result, the off-diagonal elements contain linear combinations of multiple constants, which precludes the guarantee of pairwise balancing.}}

\subsection*{Optimality of the Topologically Mutually Consistent Spin Order}

The examples above raise a fundamental question regarding the matrix representation: \textit{Is it possible to find a different spin permutation that yields a "better" or more symmetric form of the $J$-coupling matrix?} Here we demonstrate that the topologically mutually consistent ordering, utilized in the examples above, is the optimal choice.

\begin{theorem*}
For symmetric spin systems of types $[A'B']_n$ $([A'X']_n)$, the \textit{Topologically Mutually Consistent} spin ordering is most optimal. It maximizes the geometric symmetry of the $J$-coupling matrix structure by rendering the intergroup coupling block persymmetric.
\end{theorem*}

\begin{proof}
Let the condition for geometric spectral symmetry be the invariance of the $J$-matrix under reflection across the anti-diagonal.
For the block-structured matrix
\begin{equation}
    \boldsymbol{J} = 
    \begin{pmatrix}
     \boldsymbol{J}_{AA} & \boldsymbol{J}_{AX} \\
     \boldsymbol{J}_{AX}^T & \boldsymbol{J}_{XX}
    \end{pmatrix}
\end{equation}
this requires the off-diagonal block $\boldsymbol{J}_{AX}$ to be \textit{persymmetric} (symmetric with respect to its own anti-diagonal).

\paragraph{1. Analysis of the Mutually Consistent Ordering}
When the order is chosen to be topologically mutually consistent (mapping the indices of group $X$ to group $A$ via the operations of the molecular point group), the sub-block $\boldsymbol{J}_{AX}$ inherits the symmetry structure of the molecule.
For the considered symmetry groups (e.g., $C_n$, $C_{nv}$, $D_n$), this ordering ensures that the coupling constants satisfy the condition $J_{i,j} = J_{N+1-j, N+1-i}$.
Thus, the block $\boldsymbol{J}_{AX}$ becomes inherently \textit{persymmetric}, guaranteeing that the heteronuclear interactions satisfy the spectral symmetry requirements identically (automatically).

\paragraph{2. Analysis of Inconsistent Orderings}
Any permutation that violates the mutual consistency (e.g., by reordering group $X$ differently from group $A$) breaks the symmetry correspondence between the indices.
Consequently, the matrix elements symmetric with respect to the anti-diagonal would correspond to physically different coupling constants.
This introduces new asymmetry into the $J$-coupling matrix that is not intrinsic to the system's physics.

\paragraph{Conclusion}
The Topologically Mutually Consistent Ordering is unique in that it leverages the molecular symmetry to render the $\boldsymbol{J}_{AX}$ block perfectly symmetric.
It proves that the observed spectral asymmetry arises exclusively from the irreducible imbalance between the intragroup couplings ($\boldsymbol{J}_{AA} \neq \boldsymbol{J}_{XX}^{refl}$).
\end{proof}

\section{Implications for Structure Elucidation (The Inverse Problem)}

Since the General Theorem of spectrum symmetry provides necessary and sufficient conditions, the experimental observation of a mirror-symmetric spectrum imposes strict constraints on the topology of the underlying spin system.
This allows for the partial solution of the "inverse problem": inferring spin system properties solely from spectral symmetry.

Based on the General Theorem, we can formulate the following rule: \textit{If the experimental spectrum possesses mirror symmetry, then there must exist a specific spin order in which simultaneously:}

\begin{enumerate}
    \item The spin resonance frequencies are symmetrically ordered about the mid-resonance frequency ($\nu_0$); and
    
    \item In this spin order, the spectrum is invariant with respect to the reflection of the $J$-coupling matrix relative to the anti-diagonal (due to the equality or balance of the corresponding coupling constants).
\end{enumerate}

Consequently, any structural candidate that cannot support such a "palindromic" ordering of frequencies and balanced condition for couplings (e.g., due to the lack of necessary internal symmetry) can be immediately rejected.

\section{Conclusion}

The mirror symmetry of NMR spectra is governed by the rigid interplay between the distribution of resonance frequencies and the topology of the $J$-coupling network.
We have derived the \textit{General Theorem of Spectral Symmetry}, establishing that:

\begin{enumerate}
    \item \textbf{The Fundamental Criterion:} Mirror symmetry arises if and only if the spin system admits a specific numbering of nuclei (a "palindromic spin order") that \textit{simultaneously} satisfies two conditions:
    \begin{itemize}
        \item \textit{Frequency Balance:} The resonance frequencies are distributed centrosymmetrically relative to the spectral center ($\nu_i - \nu_0 = -(\nu_{N+1-i} - \nu_0)$).
        \item \textit{Interaction Invariance:} In this specific order, the spectrum remains invariant under the reflection of the $J$-coupling matrix across the anti-diagonal.
    \end{itemize}
    
    \item \textbf{Mechanism of Invariance:} This spectral invariance of the interaction is realized in two distinct ways:
    \begin{itemize}
        \item \textit{Geometric Symmetry:} The $J$-matrix is explicitly symmetric ($\boldsymbol{J} = \boldsymbol{J}^{refl}$), which implies the strict commutation with the generalized parity operator ($[\hat{H}, \hat{Q}] = 0$).
        \item \textit{Isospectral Symmetry:} The $J$-matrix is not symmetric, but forms "balanced pairs" of constants that render the Hamiltonian isospectral under reflection ($\text{Spec}(\boldsymbol{J}) = \text{Spec}(\boldsymbol{J}^{refl})$).
    \end{itemize}
\end{enumerate}

Systems with full magnetic equivalence (such as $A_n B_n$ or $A_n X_n$) satisfy these conditions trivially because the symmetry-breaking intra-group couplings do not affect the spectrum.
Systems like $AA'BB'$ satisfy them via algebraic isospectrality (balanced pairs) originating from the internal symmetry of the spin system.
However, in more complex systems with chemically equivalent but magnetically non-equivalent spins (like $[A'B']_n$), the specific topological balance of $J$-couplings required for isospectrality may be absent.
Our analysis of $[A'X']_n$ systems (such as 1,3,5-trifluorobenzene) reveals that high molecular symmetry alone is insufficient to guarantee spectral symmetry. In such cases, the lack of balance between intragroup
coupling constants required for isospectrality results in asymmetric spectra.

\end{document}